\documentclass[journal,12pt,onecolumn,draftclsnofoot]{IEEEtran}
\usepackage{graphicx,amsmath,float,amsthm,hyperref,subcaption,amssymb,color,mathdots,yhmath,listings,cite}
\usepackage[utf8]{inputenc}
\usepackage[english]{babel}

\newtheorem{lemma}{Lemma}
\newtheorem{theorem}{Theorem}
\newtheorem{remark}{Remark}
\addto\captionsenglish{}
\title{Localization Efficiency in Massive MIMO Systems}  
\author{Masoud Arash,~\IEEEmembership{Student member, IEEE}, Hamed Mirghasemi,~\IEEEmembership{Member, IEEE}, Ivan Stupia,~\IEEEmembership{Member, IEEE}, Luc Vandendorpe,~\IEEEmembership{Fellow, IEEE}
\thanks{The authors are with the Institute of Information and Communication Technologies, Electronics and Applied Mathematics (ICTEAM), Universit\'e catholique de Louvain, 1348 Louvain-la-Neuve, Belgium (E-mail: masoud.arash@uclouvain.be).
	
This work has been submitted to the IEEE for possible publication. Copyright may be transferred without notice, after which this version may no longer be accessible.}}
\begin{document}
	\maketitle
	\begin{abstract}
		In the next generation of wireless systems, Massive MIMO offers high angular resolution for localization. By virtue of large number of antennas, the Angle of Arrival (AoA) of User Terminals (UTs) can be estimated with high accuracy. According to Dense Multipath Component (DMC) channel model, local scatters around UTs can create different multipath signals for each antenna at the Base Station (BS). We obtain a deterministic form for the Cramer-Rao Lower Bound ($CRLB$) in a multi-user scenario when the contribution of the multipath signals is considered. We do this when the multipath signals are independent and identically distributed (i.i.d) with arbitrary distribution. Then, we redefine a localization efficiency function for a multi-user scenario and numerically optimize it with respect to (w.r.t) the number of antennas. We prove when only a subset of the available antennas is used, $CRLB$ can be minimized w.r.t which set of antennas is used. Then, an antenna selection strategy that minimizes $CRLB$ is proposed. As a benchmark, we apply the proposed antenna selection scheme to the MUltiple SIgnal Classification (MUSIC) algorithm and study its efficiency. Numerical results validate the accuracy of our analysis and show significant improvement in efficiency when the proposed antenna selection strategy is employed.
	\end{abstract}
	\section{Introduction}
	Massive MIMO systems are one of the prime candidates for the next generation of wireless systems \cite{larsson2014massive}. These systems employ a large number of antennas, which provides numerous opportunities for performance improvement of a wireless system, like increased capacity, spatial diversity, and lower latency \cite{larsson2014massive}. Interestingly, these systems offer high accuracy for different localization kinds, especially AoA and orientation of UTs \cite{li2019massive,guerra2015position}. In addition to these benefits, the use of massive antenna arrays would enable more efficient use of the time and frequency resources by enabling the simultaneous localization of more UTs.
	Various types of localization with different goals are introduced in the literature. Anchor Based schemes in which UTs' locations are estimated w.r.t an anchor are one of the most popular methods \cite{elsawy2017base}. In such approaches, methods like Received Signal Strength (RSS), Difference Time of Arrival, and AoA estimation are used to map UTs' locations. Performance of these methods is usually compared with the $CRLB$, which gives a lower bound on the estimation error for any unbiased estimator \cite{kay1993fundamentals}. 
	
	Several works studied $CRLB$ in Massive MIMO settings. For a planar antenna array, \cite{wang2012low} approximated it for a fading free channel. In \cite{shahmansoori2017position} authors derived $CRLB$ as a function of instantaneous parameters for AoA, angle of departure, delay, and orientation estimation of UTs for different scenarios when there is a dominant path either in Line of Sight (LoS) or Non-LoS. In \cite{fan2018angle} $CRLB$ for AoA and channel gain is obtained in a Massive MIMO system with the planar array for a single UT. Authors in \cite{abu2018error} approximated $CRLB$ for a single UT case of a planar array in mmwave case when multipath effects are considered. 
	
	All works in \cite{shahmansoori2017position}-\cite{abu2018error}, considered an identical channel coefficient from each UT to all the antennas at the BS. This means all antennas are fully correlated, and the antenna array has zero spatial diversity. The central hypothesis behind this common assumption is that all the UTs experience channels in which one or few dominant paths convey most of the received signal power to the BS. As a matter of fact, previous studies on localization in massive MIMO systems offer valuable insights into the information that can be extracted from the dominant components of the channel (if there is any), with the apparent consequence that if those components are shadowed, the $CRLB$ of systems may grow indefinitely. It is worth noting that this assumption might contradict the original idea of developing massive MIMO technologies as an efficient solution to provide seamless and reliable links between UTs and BS even in the absence of LoS or clear dominant paths. This consideration is corroborated by many studies of Massive MIMO systems, such as \cite{bjornson2015optimal}, assuming that UTs have independent channel coefficients for different antennas.     
	
	Besides, the same infrastructure might be used for both localization and data transmission. In this case, to avoid single spatial diversity, which results from fully correlated channels, antennas are placed in a way to gain independence. For instance, in mmwave scenarios, a separation of few centimeters can do this. Moreover, based on the DMC channel model \cite{astely1999effects,liu2012cost}, different channel coefficients for each antenna can stem from local scatters in the vicinity of the UTs (Fig.~\ref{sysMod}). In this model, each antenna receives a dominant path signal accompanied by multiple multipath signals. This implies that there is a clear discrepancy between studies for data transmission and localization in massive MIMO systems. The question that arises here is how $CRLB$ changes when different antennas have different channel coefficients? Or in other words, can we exploit the presence of multipath signals in massive antenna arrays to extract the AoA information and show the contribution of these multipath signals in the $CRLB$?
	
	In this perspective, \cite{camargo2018probability} tackled the problem of i.i.d channel coefficients for AoA estimation for the first time. However, due to mathematical complications of $CRLB$ analysis, \cite{camargo2018probability} only addresses the probability of AoA detection for a single UT. The first objective of this work is then to fill this gap by proposing a deterministic expression of the CRLB for multiple UTs under the hypothesis of i.i.d. channel coefficients between antennas. To obtain a deterministic $CRLB$, various ideas have been proposed to remove the effects of instantaneous nuisance parameters (e.g., fading channel coefficients) in $CRLB$. In \cite{miller1978modified}, Miller and Chang introduced a performance metric obtained by taking the expectation from the $CRLB$ w.r.t. the nuisance parameter, while in \cite{d1994modified} the authors defined a Modified $CRLB$ ($MCRLB$) by taking an expectation from Fisher Information Matrix (FIM). The downside of those proposals is that the proposed metrics depend on the particular channel probability distribution. This problem is worsened in multi-user (MU) MIMO systems. To the best of our knowledge, this is the first work proposing a deterministic and closed-form solution for the $CRLB$ in MU Massive MIMO systems when the contribution of multipath signals with arbitrary distribution is considered. 
	
	To achieve this, we take advantage of Random Matrix Theory (RMT) to prove that the $CRLB$ of a MU Massive MIMO system almost surely converges to a deterministic function of the channel variance for all possible distributions of the channel coefficients. We also show that $CRLB$ for AoA estimation always converges toward a finite value, meaning that by virtue of considered multipath signals, AoA information can always be extracted, even when the dominant path is in poor condition. This result is of particular importance to give a theoretical foundation to those techniques, such as those proposed in \cite{li2019massive,mahler2016tracking,Sibren2020NN,zhao2017tone,wielandt2017indoor}, aiming at exploiting multipath signals to extract or refine AoA estimation.
	
	Though Massive MIMO technology may guarantee seamless and reliable localization for multiple UTs with limited time and frequency resources, those benefits may be jeopardized by the increased energy consumption of those systems. The energy efficiency concern in Massive MIMO systems has drawn many research interests during last years \cite{bjornson2015optimal,andrews2014will}. In \cite{bjornson2015optimal}, authors discussed how a comprehensive model for such systems should consider the energy consumption of different parts, including hardware and signal processing units. Some of these parts that include computational and hardware energies scale with the number of antennas. Accordingly, several works have studied how performance criteria change by considering such a comprehensive model \cite{gao2015massive,arash2017employing}. Efficiency in localization is only studied in few works and mainly at the network level. In \cite{reich2008comparing} authors discussed the product of error and power consumption of a wireless sensor network as an efficiency parameter. \cite{lieckfeldt2009characterizing} used the inverse of this product to give a physical sense to this criterion in the same settings, obtaining $CRLB$ through simulations. Yet, the concept of efficiency in localization demands more attention as it can reflect essential trade-offs. 
	
	For this reason, in the second part of this work, we redefine a Localization Efficiency ($LE$) function so it can be used for extensive studies in MU scenarios. First, $LE$ is formulated with fundamental performance metrics, using obtained $CRLB$ for a typical system, number of UTs, and total energy consumption. Contrary to previous studies, we use a comprehensive energy consumption model. Interestingly, the study of $CRLB$ reveals that when a subset of available antennas is used, both the behavior and formulation of the $CRLB$ change depending on which set of antennas is utilized. Next, we study the antenna selection and find the set that minimizes $CRLB$ when only a subset of available antennas is used. Also, we show that the optimal number of antennas is changed for various antenna selection strategies. Finally, to analyze $LE$ and antenna selection in simpler system models, the MUSIC algorithm is studied. The contributions of this paper are summarized as follows:
	
	\begin{itemize}
		\item CRLB for AoA estimation of a MU Massive MIMO system is derived in a deterministic form under i.i.d channel model with unknown distribution, using RMT methods. 
		\item Efficiency function for localization is redefined as a function of system parameters for the evaluation of localization methods, and it is used to study the trade-off between performance and energy consumption.
		\item Antenna selection for localization is introduced, and a selection strategy that minimizes $CRLB$ is presented. $LE$ is reformulated in this case, and the optimal number of antennas is obtained for this selection method. 
		\item $LE$ of MUSIC algorithm based on its exact required computations is derived. Also, different antenna selection methods are studied for this algorithm. 
	\end{itemize}
	
	The remainder of this paper is organized as follows. In Section~\ref{SystemModel} we introduce our system model. $CRLB$ is calculated for different channel models of BS antennas in Section~\ref{crlb}. $LE$ is formulated in Section~\ref{LocalizationEfficiency}. All of these are then used to study the idea of antenna selection in Section~\ref{Antennaselection}. $LE$ of the MUSIC algorithm is dealt with in detail in Section~\ref{music}. In Section~\ref{Numerical}, numerical results are used to validate the theoretical analysis and make comparisons of $LE$ under various scenarios. Finally, the major conclusions are drawn in Section~\ref{Conclusion}.
	
	\emph{Notation}: Boldface lower case is used for vectors, $\boldsymbol{x}$, and upper case for matrices, $\boldsymbol{X}$. $\boldsymbol{X}^*$, $\boldsymbol{X}^T$, $\boldsymbol{X}^H$ and $\boldsymbol{X}_{k,k}$ denote conjugate, transpose, conjugate transpose and $(k,k)$th entry of $\boldsymbol{X}$, respectively. $\mathbb{E}\{.\}$ denotes expectation, $Card(.)$ is cardinality of a set, $j=\sqrt{-1}$, $\mid.\mid$ stands for absolute value of a given scalar variable, $tr$ is trace operator, $\odot$ is Hadamard product operator, $\xrightarrow{a.s.}$ means Almost Sure convergence and $diag(\boldsymbol{x})$ is a diagonal matrix whose diagonal entries are the elements of $\boldsymbol{x}$. Also, $\boldsymbol{I}_K$ is $K\times K$ identity matrix. When $\boldsymbol{y}=[y_1\hspace{2mm}y_2\hspace{2mm} \ldots\hspace{2mm} y_p]^T$ and $\boldsymbol{x}=[x_1\hspace{2mm}x_2\hspace{2mm} \ldots\hspace{2mm} x_q]^T$, we define 
	\begin{equation}
		(\frac{\partial \boldsymbol{y}}{\partial \boldsymbol{x}})_{p,q}=\frac{\partial \boldsymbol{y}_p}{\partial \boldsymbol{x}_q}. \label{JHH}
	\end{equation}
	\section{System Model} \label{SystemModel}
	We consider the uplink of a single-cell mmwave communication MU-Massive MIMO system with a BS at the center of the cell, equipped with $M$ antennas, equally separated by a distance $d$ (Fig.~\ref{sysMod}). There are $K$ single antenna UTs distributed all over the cell. In this system, BS estimates AoA of UTs and channel coefficients using the pilot signals transmitted by UTs with wavelength $\lambda$. The first antenna at the top of the antenna array is the reference point w.r.t. which the AoA is measured. In the MU scenario, UTs transmit their pilot signals at the same time and frequency.
	
	\begin{figure}[t]
		\centering
		\includegraphics{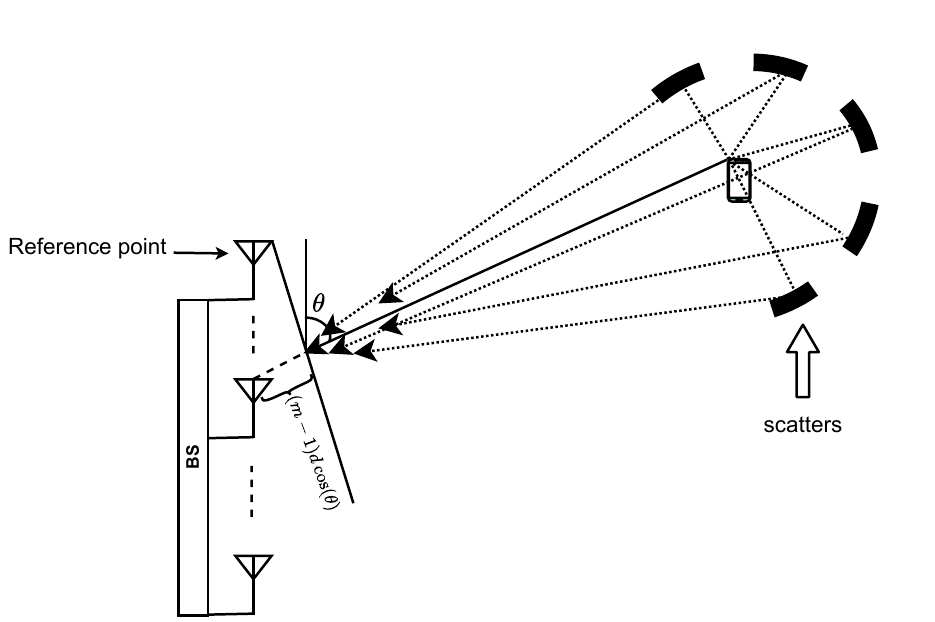}
		\caption{DMC channel model with multipath signals in the vicinity of the dominant path.}
		\label{sysMod}
	\end{figure}
	In this model, there are local scattering objects in the vicinity of the users as illustrated in Fig.\ref{sysMod}, resulting in the DMC channel model\cite{astely1999effects}. There is a dominant path that defines the AoA, showed with the solid line in Fig.~\ref{sysMod}, accompanied by many multipath signals that are originated by scatters at the vicinity of the user, depicted by dotted lines. Such a channel can be modeled as a random variable with a mean equal to the channel coefficient of the dominant path and a variance that accounts for the sum of all the multipath signals at each antenna.
	
	As the multipath signals are in the vicinity of the dominant path for each antenna, they will also travel the dashed line path, with a little difference depending on the effects of the scatters. So, the phase of the multipath signals at the $m$th antenna can be modeled as a random variable, with the mean equal to $2\pi\frac{(m-1)d\cos(\theta)}{\lambda}$ (that is the phase difference created due to the travel of signal in the dashed line). The total channel coefficient for the $m$th antenna from the $k$th UT can be written as
	\begin{align}
		g_{m,k}&=l(r_k)(h_{d_k}e^{-j\frac{2\pi r_k}{\lambda}}e^{-j\frac{2\pi}{\lambda}d(m-1)\cos(\theta_k)}+\check{h}_{r_{m,k}}e^{-j\alpha_{m,k}}e^{-j\frac{2\pi}{\lambda}d(m-1)\cos(\theta_k)})\nonumber \\
		&=l(r_k)e^{-j\frac{2\pi}{\lambda}d(m-1)\cos(\theta_k)}(\bar{h}_k+\hat{h}_{r_{m,k}})=l(r_k)e^{-j\frac{2\pi}{\lambda}d(m-1)\cos(\theta_k)}h_{m,k},
	\end{align}
	where $l(r_k)$ is the large scale fading of the $k$th UT, $h_{d_k}$ is the amplitude of the dominant path, $r_k$ is the length of the dominant path, $\check{h}_{r_{m,k}}$ and $\alpha_{m,k}$ are the amplitude and phase of aggregated multipath signals at the $m$th antenna. Also, $\bar{h}_k=h_{d_k}e^{-j\frac{2\pi r}{\lambda}}$ and $\hat{h}_{r_{m,k}}=\check{h}_{r_{m,k}}e^{-j\alpha_{m,k}}$. Finally, $h_{m,k}$ is a random variable with a mean equal to $\bar{h}_k$ and variance equal to $\sigma^2_h$ that accounts for the power of the multipath signals from the $k$th user at the $m$th antenna.
	Therefore, the received signal at the BS side can be written as
	\begin{equation}
		\boldsymbol{y}=\boldsymbol{G}\boldsymbol{s}+\boldsymbol{n}, \label{recived signal}
	\end{equation}
	in which $\boldsymbol{G}$ is $M\times K$ channel matrix between $M$ BS antennas and $K$ UTs, $\boldsymbol{s}\in\mathbb{C}^{K\times 1}$ is the vector of transmitted pilots and $\boldsymbol{n}\sim \mathcal{CN}(0,\sigma_n^2\boldsymbol{I}_M)$ is additive noise. Matrix $\boldsymbol{G}$ is composed as
	\begin{equation}
		\boldsymbol{G}=(\boldsymbol{A}_{Rx}\odot \boldsymbol{H})\boldsymbol{B}^{\frac{1}{2}}, \label{G composition}
	\end{equation}
	where
	\begin{equation}
		\boldsymbol{A}_{Rx}=[\boldsymbol{a}_{Rx}(\theta_1)\hspace{2mm} \boldsymbol{a}_{Rx}(\theta_2)\hspace{2mm} \ldots \hspace{2mm}\boldsymbol{a}_{Rx}(\theta_K)]
	\end{equation}
	contains $M\times 1$ steering vectors of BS antenna array response for $K$ UTs in which 
	\begin{equation}
		\boldsymbol{a}_{Rx}(\theta_k)=\frac{1}{\sqrt{M}}[1 \hspace{2mm}e^{-j\beta\cos(\theta_k)}\hspace{2mm} \ldots \hspace{2mm}e^{-j(M-1)\beta\cos(\theta_k)}]^T, \label{aTx}
	\end{equation}
	$\theta_k$ is $k$th UT's AoA for $k\in\{1,2,\ldots,K\}$ and $\beta=\frac{2\pi d}{\lambda}$. 
	$\boldsymbol{H}$ is an $M\times K$ matrix whose $(m,k)$th element, $h_{m,k}$, is fast fading coefficient between $k$th UT and $m$th BS antennas
	\begin{equation} 
		h_{m,k}=h^r_{m,k}+jh^i_{m,k}, \label{Hiid}
	\end{equation}
	and
	\begin{align} 
		&\mathbb{E}\{h^r_{m,k}\}=\mathcal{R}e\{h_{d_k}\}, \hspace{2mm} \mathbb{E}\{h^i_{m,k}\}=\mathcal{I}m\{h_{d_k}\},\nonumber\\
		&\mathbb{E}\{|h^r_{m,k}|^2\}-(\mathbb{E}\{h^r_{m,k}\})^2=\mathbb{E}\{|h^i_{m,k}|^2\}-(\mathbb{E}\{h^i_{m,k}\})^2=\sigma_h^2, \label{RIofH}
	\end{align}
	for $m\in\{1,2,\ldots,M\}$ and $k\in\{1,2,\ldots,K\}$. $\boldsymbol{B}$ is an $K\times K$ diagonal matrix whose $k$th diagonal element is $l(r_k)$. Also, for simplicity, we define received signal to noise ratio at the BS side as
	\begin{equation} 
		\rho_k\triangleq\frac{|s_k|^2l(r_k)}{\sigma_n^2}. \label{snr}
	\end{equation}
	\section{CRLB} \label{crlb}
	In this section, we derive a deterministic formula for the $CRLB$ of AoA estimation for MU Massive MIMO settings where multipath signals' contribution is considered. We assume that the channel coefficients are random but known and so not needed to be estimated. Vector of desired parameters will be
	\begin{equation}
		\boldsymbol{\eta}_\theta=[\theta_1 \hspace{2mm} \theta_2\hspace{2mm}\ldots \hspace{2mm}\theta_K]^T. \label{etatet}
	\end{equation}
	Defining $\hat{\boldsymbol{\eta}}_\theta$ as the unbiased estimation of $\boldsymbol{\eta}_\theta$, the Mean Square Error (MSE) of the estimator is lower bounded as \cite{kay1993fundamentals}
	\begin{equation}
		\mathbb{E}_{\boldsymbol{y}|\boldsymbol{\eta}_\theta}\{(\boldsymbol{\eta}_\theta-\boldsymbol{\hat{\eta}}_\theta)(\boldsymbol{\eta}_\theta-\boldsymbol{\hat{\eta}}_\theta)^T\}\ge \boldsymbol{CRLB}_{\theta}=\boldsymbol{J}^{-1}_{\theta,\theta}, \label{mse}
	\end{equation}
	where $\boldsymbol{J}$ is FIM and is defined as \cite{kay1993fundamentals}
	\begin{equation}
		\boldsymbol{J}_{\theta,\theta}=\mathbb{E}_{\boldsymbol{y}|\boldsymbol{\eta}_\theta}[-\frac{\partial^2 \ln f(\boldsymbol{y}|\boldsymbol{\eta}_\theta)}{\partial \boldsymbol{\eta}_\theta \partial \boldsymbol{\eta}_\theta^T}], \label{FIM def}
	\end{equation}
	where $f(\boldsymbol{y}|\boldsymbol{\eta})$ is the likelihood function of the received signal. From \cite{kay1993fundamentals} (page 525 and proof on page 563)
	\begin{equation}
		\boldsymbol{J}_{\theta,\theta}=\frac{2}{\sigma_n^2}\mathcal{R}e[(\frac{\partial \boldsymbol{w}}{\partial \boldsymbol{\eta}_\theta})^H\frac{\partial \boldsymbol{w}}{\partial \boldsymbol{\eta}_\theta}], \label{FIM1}
	\end{equation}
	where
	\begin{equation}
		\boldsymbol{w}\triangleq\boldsymbol{G}\boldsymbol{s}=\frac{1}{\sqrt{M}}
		\left[\begin{array}{c} 
			\sum_{i=1}^{K}h_{1,i}s_il(r_i) \vspace{2mm}\\
			\sum_{i=1}^{K}h_{2,i}s_il(r_i)e^{-j\beta\cos(\theta_i)}\\
			\vdots\\
			\sum_{i=1}^{K}h_{M,i}s_il(r_i)e^{-j(M-1)\beta\cos(\theta_i)}\\
		\end{array}\right]. \label{w}
	\end{equation}
	By taking derivative w.r.t $\theta$, we have
	\begin{equation}
		(\frac{\partial \boldsymbol{w}}{\partial \boldsymbol{\eta}_\theta})_{m,k}=\boldsymbol{X}_{m,k}=\frac{j\beta(m-1)\sin(\theta_k) h_{m,k}s_kl(r_k)e^{-j(m-1)\beta\cos(\theta_k)}}{\sqrt{M}}. \label{dw}
	\end{equation}
	for $m\in\{1,2,\ldots,M\}$ and $k\in\{1,2,\ldots,K\}$. So, using Eq.~\ref{mse}, $\boldsymbol{CRLB}_\theta$ will be
	\begin{equation}
		\boldsymbol{CRLB}_\theta=\boldsymbol{J}^{-1}_{\theta,\theta}=\frac{\sigma_n^2}{2}(\mathcal{R}e(\boldsymbol{X}^H\boldsymbol{X}))^{-1}. \label{NCRLBcos}
	\end{equation}
	By virtue of the independence of antennas, using the following lemmas from RMT, we prove that for any distribution of $\boldsymbol{H}$, $\boldsymbol{CRLB}_{\cos(\theta)}$ almost surely converges toward a deterministic closed-form expression that is a function of system parameters, such as the number of antennas and variance of channel coefficients.
	\begin{lemma} \label{Lem1}
		Let $\boldsymbol{\Sigma} \in \mathbb{C}^{N\times N}$, be a matrix with uniformly bounded spectral norm. Let $\boldsymbol{x} \in \mathbb{C}^N$, be a random vector with i.i.d. entries of zero mean, variance $\frac{1}{N}$ and eighth order moment of order $O(\frac{1}{N^4})$, independent of $\boldsymbol{\Sigma}$. Then
		\begin{equation}
			\boldsymbol{x}^H\boldsymbol{\Sigma}\boldsymbol{x}-\frac{1}{N}tr(\boldsymbol{\Sigma})\xrightarrow{a.s.}0,
		\end{equation}
		as $N\rightarrow \infty$.
	\end{lemma}
	\begin{proof}
		See \cite{couillet2011random}.
	\end{proof}
	\begin{lemma} \label{Lem2}
		For $\boldsymbol{\Sigma} \in \mathbb{C}^{N\times N}$, be a matrix with uniformly bounded spectral norm, $\boldsymbol{x}$ and $\boldsymbol{y}$ two vectors of i.i.d. variables such that $\boldsymbol{x} \in \mathbb{C}^N$ and $\boldsymbol{y} \in \mathbb{C}^N$ have zero mean, variance $\frac{1}{N}$ and fourth order moment of order $O(\frac{1}{N^2})$, we have
		\begin{equation}
			\boldsymbol{x}^H\boldsymbol{\Sigma}\boldsymbol{y}\xrightarrow{a.s.}0.
		\end{equation}
	\end{lemma}
	\begin{proof}
		See \cite{couillet2011random}.
	\end{proof}
	\begin{lemma} \label{Lem3}
		For large $M$, when $\cos(\theta_1)\ne\cos(\theta_2)$, for $\boldsymbol{\Sigma}=diag(\frac{1}{M^2},\frac{4}{M^2},\ldots,1)$, we have 
		\begin{equation}
			\mathcal{R}e\{\boldsymbol{a}_{Rx}(\theta_1)^H\boldsymbol{\Sigma}\boldsymbol{a}_{Rx}(\theta_2)\}\xrightarrow{a.s.}0. \label{asd}
		\end{equation}
	\end{lemma}
	\begin{proof}
		See Appendix \ref{Appl3}.
	\end{proof}
	Using these lemmas and independence of the channel distribution from the number of antennas at the BS, the following theorem gives the deterministic expression for $\boldsymbol{CRLB}_\theta$.
	\begin{theorem} \label{T1}
		In a Massive MIMO system with large number of antennas, $\boldsymbol{CRLB}_\theta$ converges toward a deterministic form as
		\begin{equation}
			\boldsymbol{CRLB}_{\theta}\xrightarrow{a.s.}\frac{3}{\beta^2(M-1)(2M-1)}\boldsymbol{S}, \label{CRLBAStheta}
		\end{equation}	
		in which $\boldsymbol{S}$ is an $K\times K$ diagonal matrix with
		\begin{equation}
			\boldsymbol{S}_{k,k}=((|h_{d_k}|^2+2\sigma_h^2)\rho_k\sin^2(\theta_k))^{-1}, \hspace{3mm} k\in\{1,\ldots, K\}.
		\end{equation}
	\end{theorem}
	\begin{proof}
		See Appendix \ref{App1}.
	\end{proof}
	
	It is seen from Eq.~\ref{CRLBAStheta} that in a Massive MIMO system, regardless of channel distribution, instantaneous $CRLB$ tends toward a deterministic value. Although other definitions like  Miller-Chang version \cite{miller1978modified} or $MCRLB$ \cite{d1994modified} obtain a deterministic form for $CRLB$, they use expectation that requires the knowledge of channel coefficients' distribution. By contrast, our obtained expression requires only the knowledge of channel coefficients' variance and is applicable even when the distribution of channel coefficients is unknown.
	
	$CRLB$ in a Massive MU-MIMO system is inversely related to the second-order of the number of antennas that is in accordance with previous analysis of these systems in \cite{fan2018angle,abu2018error}. When the number of antennas is large enough, the effects of the multipath signals that accompany the dominant path signal are in terms of general channel statistics instead of their individual realization, allowing designers to understand the amount of information they can recover from the channel. Also, the $\sigma_h^2$ term in Eq.~\ref{CRLBAStheta} prevents the $CRLB$ from growing indefinitely when the dominant path-signal is in poor condition. This is because as the number of antennas grows, it is almost surely improbable that all antennas be in poor condition simultaneously. So, even if few antennas have a low fading coefficient, not all of the information in the system is lost, and AoA can still be extracted. The same phenomenon also happens for channel capacity in these systems \cite{bjornson2015optimal}.
	
	Several works have recently attempted to use path-loss signals with different channel coefficients to improve their estimation accuracy, such as \cite{li2019massive,mahler2016tracking}. One of the main uses of multipath signals is in Fingerprinting localization method\cite{li2019massive,steiner2011efficient}. In this method, the area in which localization is going to be done is divided into separate subsections. Each subsection is distinguished from others based on the channel that BS sees from that subsection. Here, if the path-loss signals are considered, subsections can be sharply distinguished from one another. Moreover, by considering the multipath signals In \cite{Sibren2020NN} authors extracted location information (including AoA) with weak dominant path signal, using neural networks. In \cite{zhao2017tone} authors presented a method to extract AoA in i.i.d. Rician channels. These works are examples of how the multipath signals can be used to extract or refine the AoA information. Our results are in accordance with these works and provide the theoretical background that $CRLB$ in such scenarios is finite, and estimating AoA is possible.
	
	In the rest of the paper, we focus on the system's energy efficiency in the AoA estimation part and how the utilization of antennas can affect the efficiency. This is important when one wants to use a subset of total available antennas for AoA estimation.
	
	\section{Localization Efficiency} \label{LocalizationEfficiency}
	In this section, we formulate $LE$ function, which will aggregate the benefits and costs of a localization method. To make a general criterion for different kinds of localization, we use parameters that are included in most of the localization methods. On the one hand, these parameters are the number of simultaneously localized UTs and the method's localization accuracy as benefit parameters. On the other hand, the system's total energy consumption in the localization phase is a cost parameter.
	\subsection{Accuracy Function}
	Accuracy is one of the major evaluation parameters in localization \cite{li2013robust,sallouha2019localization}. Different works have studied the accuracy function of localization methods and optimized it w.r.t. various parameters. Generally, accuracy is defined as the trace of the inverse of equivalent FIM \cite{li2013robust,shen2010fundamental}, or as the inverse of square root, the trace of the $CRLB$ matrix \cite{sallouha2019localization}. As we have achieved a deterministic expression for $CRLB$, we use the inverse of the square root of its trace for the accuracy function 
	\begin{equation}
		Accuracy=\frac{1}{\sqrt{tr(\boldsymbol{CRLB}_\theta)}}. \label{acc}
	\end{equation}
	\subsection{Energy Consumption}
	Nowadays, one of the critical parameters of a wireless system is energy consumption \cite{bjornson2015optimal}. Due to growing concerns about energy, designers have to carefully consider their systems' energy consumption and include it in system characterization. Analyzing the energy consumption of a system indicates at what cost a performance is obtained. For example, wireless systems' energy efficiency has been widely used to describe performance trade-off between rate and energy consumption \cite{bjornson2015optimal,arash2017employing}. 
	
	In order to conduct a comprehensive investigation, it is of paramount importance to consider the energy consumption of all parts of a wireless system. In addition to transmitted power, in a Massive MIMO system, the energy consumption of system hardware should also be considered to obtain a comprehensive model \cite{bjornson2015optimal}. For instance, the energy consumption of antennas' RF-chains and processing units that scale with the number of antennas is not negligible in Massive MIMO systems. In the following, we investigate different parts of the total energy consumption function according to our system model.
	
	\subsubsection{Transmitted Energy}
	In the uplink of a wireless system, UTs transmit pilots to become localized by the BS. Usually, this energy is vital as UTs have a limited energy budget. This is one of the main reasons which prevents broad utilization of methods like GPS in the next generation of wireless systems. Pilot signals are predefined with a certain energy. This energy is linearly related to the number of UTs. Therefore, transmitted energy will be
	\begin{equation}
		E_{tr}=\zeta W \frac{tr(\boldsymbol{s}\boldsymbol{s}^H)}{\omega}=\frac{W\zeta}{\omega} \sum_{i=1}^{K}|s_i|^2\hspace{2mm}(J), \label{Ptr}
	\end{equation}
	in which $\zeta$ and $W$ are duration and bandwidth of transmitted pilots, respectively, and $\omega\in(0,1]$ is UT's RF amplifier efficiency (and it is constant during pilot transmission time) \cite{bjornson2014massive}.
	\subsubsection{Processing Energy}
	We assume that the BS carries out all of the required processing. As far as LE is concerned, these processes include the detection of pilots and running localization algorithm. This energy is proportional to the number of operations, which in turn is a function of system parameters such as $M$ and $K$. To evaluate this part's energy consumption, one needs to calculate the number of required operations of an algorithm and the computational efficiency of BS processing hardware. Generally, Maximum-likelihood (ML) method that obtains $CRLB$ has the calculation complexity of $K^M$ \cite{gurbuz2020crlb}. This is the worst case, and there may be ways to reduce the number of required calculations. Sub-optimum algorithms have $M^3$ order of complexity, at most, but they do not necessarily obtain $CRLB$. We study one of these algorithms in section \ref{music}. Hardware computational efficiency, $L_{BS} (FLOP/J)$, is usually expressed as the number of Floating Point Operations (FLOP) per Second per Watt that hardware consumes \cite{bjornson2015optimal}. So, assuming the same time as pilot transmission is used for processing, the processing energy consumption, $E_p$, for ML will be formulated as 
	\begin{equation}
		E_p=\frac{K^M}{L_{BS}}W\zeta \hspace{3mm}(J). \label{pp}
	\end{equation}
	\subsubsection{Hardware Energy Consumption}
	Generally, the hardware of a wireless system can be divided into two parts: 
	\begin{itemize}
		\item Infrastructure part of a system that includes backhaul systems and network part, static circuit energy consumption and so on, which use constant energy and are necessary for system maintenance. Energy consumption of this part is independent from the number of antennas at the BS but may depend on the number of UTs \cite{bjornson2014massive}.
		\item RF-chains of BS and UTs' antennas that are proportional to the number of BS antennas and number of UTs.
	\end{itemize} 
	Therefore, hardware energy consumption, $E_h$ will be
	\begin{equation}
		E_h=\zeta(MP_{BS}+KP_{UT}+P_{fix})\hspace{3mm}(J), \label{Ph}
	\end{equation}
	where $P_{BS}$ and $P_{UT}$ are the power consumption of the hardware of each antenna (like RF-chains) in BS and UTs, respectively, and $P_{fix}$ accounts for all powers that are not related to $M$. 
	
	Therefore, total energy consumption will be
	\begin{equation}
		E_t=E_{tr}+E_p+E_h. \label{Pt}
	\end{equation}
	Finally we can formulate $LE$ as
	\begin{equation}
		LE=\frac{K(Accuracy)}{Energy Consumption}=\frac{K}{(E_t)(\sqrt{tr(\boldsymbol{CRLB}_\theta)})}. \label{LE1}
	\end{equation}
	Replacing all equations by their formula, LE will be obtained as
	\begin{equation} 
		LE=\frac{K\sqrt{\beta^2(M-1)(2M-1)}}{\sqrt{3}\zeta(\frac{W}{L_{BS}}K^M+MP_{BS}+KW \sum_{i=1}^{K}\|s_i\|^2+KP_{UT}+P_{fix})\sqrt{tr(\boldsymbol{S})}}. \label{LEp1}
	\end{equation}
	In this equation, $LE$ is a function of several system parameters, such as $M$, $K$, and $\theta$. Although designers usually cannot control some parameters such as UT's AoA or their number, they do have access to the number of BS antennas. As $LE$ reflects a trade-off between accuracy and energy consumption, this creates an opportunity that can be used to design a system that operates at the optimal point of this trade-off. For example, a common scenario that happens, in reality, is that number of UTs varies drastically at different times, and it is not necessarily efficient to use all of the antennas for any number of UTs. In this case, to maximize LE, BS should use a portion of the total available antennas.
	
	In order to optimize Eq.~\ref{LEp1} w.r.t. the number of antennas, first, the method that will select the utilized antennas should be devised. This is because how the antennas are selected changes the formulation of LE and consequently changes the number of optimal antennas. Therefore, first we study if the optimal number of antennas is fewer than the total available antennas, what is the optimal antenna selection strategy. Next, with the help of an optimal selection strategy, we achieve the formulation of LE when antenna selection is employed. Then, we can optimize the number of optimal antennas.
	
	\section{Antenna selection} \label{Antennaselection}
	In this section, we analyze the effect of antenna selection in a Massive MIMO system for localization. When the number of utilized antennas is smaller than total available antennas, e.g., due to $LE$ optimization or fewer available RF-chains, there is an opportunity that if a specific set of antennas are deployed, $LE$ can be improved even more. Energy consumption of the system in Eq.~\ref{Pt} is only a function of the number of antennas and is independent from which antennas are being used. On the other hand, in addition to the number of antennas, $\boldsymbol{CRLB}_\theta$ is a function of the set of antennas that are being used. To show this, we recall Eq.~\ref{sdsd} that states $\boldsymbol{CRLB}_\theta$ is inversely proportional to $tr(\boldsymbol{\Sigma})$. Assuming optimal number of antennas (or available RF-chains) is $F$, to minimize $\boldsymbol{CRLB}_\theta$, $tr(\boldsymbol{\Sigma})$ has to be maximized by choosing a subset of utilized antennas $\mathcal{S}$. It should be noted that in this case, the $\frac{1}{M}$ and $\frac{1}{M^2}$ coefficients in Eq.~\ref{a+b} and Eq.~\ref{Sig}, are changed to $\frac{1}{F}$ and $\frac{1}{F^2}$, respectively. So, we have
	\begin{align}
		&\max_{\mathcal{S}\subset \{1,\ldots, M\}}\sum_{x\in \mathcal{S}}\frac{(x-1)^2}{F^2} \label{opt} \\& s.t. \hspace{2mm} Card(\mathcal{S})=F.\nonumber 
	\end{align}
	Optimal solution for this problem consists of the last $F$ antennas 
	\begin{equation}
		\mathcal{S}^*=\{(M-F),\ldots, M-1\}, \label{Sopt}
	\end{equation}
	that results in the maximum value for the trace as
	\begin{align}
		\sum_{x=M-F}^{M-1}\frac{x^2}{F^2}=\frac{6M(M-F-1)+(F+1)(2F+1)}{6F}. \label{maxtr}
	\end{align}
	and minimum $\boldsymbol{CRLB}_\theta$ for $\mathcal{S}^*$ as
	\begin{align}
		\boldsymbol{CRLB}_{\theta}^*\xrightarrow{a.s.}\frac{3}{\beta^2(6M(M-F-1)+(F+1)(2F+1))}\boldsymbol{S}. \label{CRLBASthetaopt}
	\end{align}
	The interpretation of Eq.~\ref{Sopt} is that if for any reason fewer antennas than available antennas should be used, the optimal choice is to start selecting antennas from the furthest antenna w.r.t. the reference point (whose location is fixed at the top of the antenna array) and move toward it. We recall this set of antennas as the \emph{furthest set}. With this approach, $\boldsymbol{CRLB}_\theta$ will be dramatically reduced relative to the case when we choose antennas from the beginning of the array. For comparison, we write $\boldsymbol{CRLB}_{\theta}$ when $F$ first antennas (\emph{first set}) are used
	\begin{equation}
		\boldsymbol{CRLB}_{\theta}\xrightarrow{a.s.}\frac{3}{\beta^2(F-1)(2F-1)}\boldsymbol{S}. \label{CRLBASthetaS1}
	\end{equation}
	
	It can be seen from Eq.~\ref{CRLBASthetaopt} and Eq.~\ref{CRLBASthetaS1} that while $\boldsymbol{CRLB}_\theta$ for the first set is only a function of $F$, it is a function of both $F$ and $M$ for the furthest set. In other words, the first set antenna selection approach does not fully appreciate the presence of a large antenna array. Furthermore, $\boldsymbol{CRLB}_\theta$ for the first set is a \emph{decreasing} function of $F$, however, $\boldsymbol{CRLB}^*_\theta$ of the furthest set will be an \emph{increasing} function of $F$. The reason for this phenomenon lies within the $M^{-\frac{1}{2}}$ normalization factor in Eq.~\ref{aTx} (which becomes $F^{-\frac{1}{2}}$ in the antenna selection scenario) that finally results in the normalization of $tr(\boldsymbol{\Sigma})$. When we start adding antennas from the end of the array, we start from an antenna with the largest contribution to the trace, $1$, and minimum normalization cost, $1$. Then, a smaller value is added, but as the summation is normalized by the cardinality of the set, it will decrease because for any $1\le i\le F$,
	\begin{equation}
		1>\frac{1+(1-i/F^2)^2}{2}. \label{Obv1}
	\end{equation}
	Therefore, the denominator of $\boldsymbol{CRLB}_{\theta}$ decreases, and in turn, $\boldsymbol{CRLB}_{\theta}$ increases. This process is reversed for the first set of antennas. In this case, we start from an antenna with the lowest contribution to the trace and add higher values as we use more antennas. So, the normalized summation is increasing because for any $1\le i\le F$,
	\begin{equation}
		\frac{1}{F^2}<\frac{1+(1+i)^2}{2F^2}. \label{Obv2}
	\end{equation}
	Moreover, the physical explanation of this antenna selection strategy can further clarify the behavior of its $\boldsymbol{CRLB}_{\theta}$. The further an antenna is from the reference point, the more its received signal differs from the signal received in the reference point. This is because it travels through a longer path and has more time to differ from the reference point's signal. On the other hand, due to the $F^{-\frac{1}{2}}$ normalization factor, the total collected power of the system is normalized with the number of utilized antennas. From $\boldsymbol{CRLB}_{\theta}$ point of view, the system prefers to collect all of the power from the antennas, which provides the maximum possible difference from the reference point. Therefore, when antennas are selected from the end of the array, the system collects its normalized received power from antennas that provide the maximum possible difference with the signal received by the reference point. Furthermore, adding more antennas with this strategy means using more antennas that are relatively closer to the reference point than the last antenna. So, some of the power is collected from the antennas with relatively less difference (due to their shorter path) compared to the last antenna. This explains why the system's performance, in terms of $\boldsymbol{CRLB}_{\theta}$, degrades when more antennas are selected.
	
	If we omit the normalization factor, $\boldsymbol{CRLB}_{\theta}^*$ will be a decreasing function of $F$, just with a different slope of $\boldsymbol{CRLB}_{\theta}$. In other words, this antenna selection strategy reduces the slope of the $CRLB$ decrease but also reduces its initial value significantly. The decrement in the initial point is large enough that for any $F<M$, from Eq.~\ref{CRLBASthetaopt} and Eq.~\ref{CRLBASthetaS1}, it is evident that $\boldsymbol{CRLB}_{\theta}^*<\boldsymbol{CRLB}_{\theta}$, proving that selecting antennas from the furthest set is always beneficial.
	
	So, depending on the method that we select operating antennas, both $CRLB$'s formula and behavior w.r.t. the number of utilized and all available antennas is changed. In this regard, using a set of antennas that are furthest from the reference point minimizes $CRLB$. 
	
	If $F$ furthest antennas in the array are selected for localization, $LE$'s formula will be changed to
	\begin{equation}
		LE_S=\frac{K\sqrt{\beta^2(6M(M-F-1)+(F+1)(2F+1))}}{\sqrt{3}\zeta(\frac{W}{L_{BS}}K^F+FP_{BS}+KW \sum_{i=1}^{K}\|s_i\|^2+KP_{UT}+P_{fix})\sqrt{tr(\boldsymbol{S})}}. \label{LEl1}
	\end{equation}
	Now that the formula of LE is obtained according to the optimal antenna selection strategy, the following theorem gives the number of optimal antennas when this antenna selection strategy is used.
	\begin{theorem} \label{T3}
		When operating antennas are selected from the furthest set, optimal number of them is
		\begin{equation}
			F^*=K+1. \label{Fopt}
		\end{equation}	
	\end{theorem}
	\begin{proof}
		Eq.~\ref{Pt} and Eq.~\ref{CRLBASthetaopt} clearly show that both $E_t$ and $CRLB$ are increasing functions of $F$. Therefore, maximum value of LE in Eq.~\ref{LEl1} happens for the minimum possible value of $F$ which is $K+1$.
	\end{proof}
	\begin{remark} 
		It should be noted that the furthest antenna selection strategy is optimal as long as the steering vector can be modeled as Eq.\ref{aTx}. In this equation, it is assumed that the array size is much smaller than the distance of UTs from it, so the difference in power of arrival is negligible. When the array becomes too large w.r.t the distance of the UTs from it, the difference in power may be so large that it affects antennas' contribution in different distances from the reference point. Also, for large size arrays, the incident wave cannot be modeled as a plane. So, the proposed antenna selection is optimal as long as the distance from the reference point to the last antenna at the BS is much smaller than the distance of UTs from the BS, and Eq.~\ref{aTx} models the steering vector.
	\end{remark}
	When the effects of the multipath signals are ignored, i.e., $\sigma_h^2\to0$, it is seen that $\boldsymbol{CRLB}_{\theta}$ is still proportional with the inverse of $tr(\boldsymbol{\Sigma})$, which means our results for antenna selection is applicable for this setting, too. Therefore, no matter what the channel model is, by using the proposed antenna selection method, $\boldsymbol{CRLB}_{\theta}$ can be minimized, and $LE$ can be further improved. In the next section, we present a case study of how the $LE$ can be formulated for an estimation algorithm.
	
	\section{MUSIC Algorithm} \label{music}
	In order to study antenna selection effects on system performance when the effects of multipath signals are not considered, we use it for one of the most known algorithms for AoA estimation, MUSIC, which is being used in several applications \cite{wang2018angle}. We study its $LE$ and how antenna selection improves it. This helps to clarify that antenna selection is beneficial, no matter what the channel model is. After a brief description of its procedure, we calculate the exact amount of calculations that are required by this algorithm and formulate its $LE$. Then in section~\ref{Numerical} we compare $LE$ when antennas are selected from the furthest and first set. 
	\subsection{Procedure}
	Consider that we have a received signal as
	\begin{equation}
		\boldsymbol{y}_F=\boldsymbol{A}_F\boldsymbol{s}+\boldsymbol{n}_F, \label{recived signalMU}
	\end{equation}
	subscript $F$ shows number of rows, as $F$ antennas are deployed. Sample covariance matrix of this signal can be obtained as \cite{zhang2010direction} 
	\begin{equation}
		\tilde{\boldsymbol{R}}_y=\frac{1}{N}\sum_{i=1}^{N}\boldsymbol{y}_{F_i}\boldsymbol{y}_{F_i}^H. \label{tildeR}
	\end{equation}
	In the Eigenvalue Decomposition (EVD) of $\tilde{\boldsymbol{R}}_y$, there will be $K$ eigenvectors corresponding to $K$ UTs and $F-K$ eigenvectors corresponding to the noise. Each noise eigenvector is orthogonal to the columns of $\boldsymbol{A}$. So, by forming $\boldsymbol{E}_n$, composed of noise eigenvectors,
	\begin{equation}
		\boldsymbol{E}_n=[\boldsymbol{v}_{K+1} \hspace{3mm}\boldsymbol{v}_{K+2}\hspace{3mm} \ldots\hspace{3mm} \boldsymbol{v}_{F}], \label{eig5}
	\end{equation}
	we can form a spatial spectrum function as \cite{stoeckle2015doa}
	\begin{equation}
		P(\theta_i)=\frac{1}{\boldsymbol{g}(\theta_i)^H\boldsymbol{E}_n\boldsymbol{E}_n^H\boldsymbol{g}(\theta_i)}, \label{Spec}
	\end{equation}
	in which $\boldsymbol{g}$ is a subset of $\boldsymbol{a}_{Rx,k}$ for the antenna subset that is being used, e.g., if the furthest set of antennas is used it will be
	\begin{align}
		\boldsymbol{g}(\theta_i)=\frac{1}{\sqrt{F}}[e^{-j(M-F)\beta\cos(\theta_i)}\hspace{2mm} \ldots \hspace{2mm}e^{-j(M-1)\beta\cos(\theta_i)}]^T, \hspace{4mm}
		\theta_i\in[0,\frac{\pi}{Q}, \ldots,\frac{\pi-1}{Q}]\label{gteta} ,
	\end{align}
	where $Q$ is search cardinality. Peaks of $P(\theta)$ happens when $\boldsymbol{g}(\theta)$ corresponds to one of the actual steering vectors. These peaks are estimated as AoAs. Therefore, steps of MUSIC are
	\begin{itemize}
		\item[i)] Observe $N$ snapshots and construct $\tilde{\boldsymbol{R}}_y$ in Eq.~\ref{tildeR}. 
		
		\item[ii)] Calculate the EVD of $\tilde{\boldsymbol{R}}_y$ and extract $\boldsymbol{E}_n$ in Eq.~\ref{eig5}. 
		
		\item[iii)] Construct $P(\theta)$ and extract its maximum points.
	\end{itemize}
	\subsection{Energy Consumption of MUSIC}
	To accurately determine the required energy consumption of the MUSIC algorithm, we first analyze the number of its required operations. With the help of \cite{boyd2004convex}, we evaluate the number of required arithmetic operations to run MUSIC. 
	
	In the first step, the algorithm multiplies an $F\times 1$ vector by its conjugate transpose $N$ times, sum them, and then divides all of its elements by $N$. Every product of the vectors requires $F^2$ operations, every matrix sum requires $F^2$ operations, and the final divide needs $F^2$ operations \cite{boyd2004convex}. So, the first step needs $(2N+1)F^2$ operations. 
	
	Generally, EVD of an $F\times F$ matrix by QR decomposition-based algorithm needs at least $F^3$
	operations \cite{boyd2004convex}. So, the second step demands $F^3$ operations. 
	
	In the last step, ignoring the search part, we need to calculate Eq.~\ref{Spec}, $Q$ times. Product of $\boldsymbol{E}_n\boldsymbol{E}_n^H$ needs $F^2(2(F-K)-1)$ operations and then, multiplying the resulted $F\times F$ matrix by two $F\times 1$ vectors from both sides, needs $(2F-1)(F+1)$ operations. Therefore, third step requires $Q[2F^2(F-K)+F^2+F-1]$ operations. 
	
	So, the total number of arithmetic operations of MUSIC is
	\begin{align}
		N_A=(2Q+1)F^3+(2N+Q(1-2K)+1)F^2+QF-Q. \label{Nofflops}
	\end{align}
	Consequently, processing energy consumption of MUSIC will be 
	\begin{equation}
		E_p=\frac{N_A}{L_{BS}}W\zeta. \label{PCMusic}
	\end{equation}
	Considering same transmitted power for all UTs, $p$, $E_{tr}$ is
	\begin{equation}
		E_{tr}=\frac{NW\zeta tr(\boldsymbol{s}\boldsymbol{s}^H)}{\omega}=\frac{NKW\zeta p}{\omega}. \label{PtrMUSIC}
	\end{equation}
	
	Hardware energy consumption will be same as Eq.~\ref{Ph}. Therefore, total energy consumption of localization process using MUSIC algorithm is 
	\begin{align}
		E_t&=E_t+E_{tr}+E_h=\underbrace{W\zeta(\frac{2Q+1}{L_{BS}})}_{C_3} F^3+\underbrace{W\zeta(\frac{2N+1+Q(1-2K)}{L_{BS}})}_{C_2}F^2\nonumber\\&+\underbrace{W\zeta(P_{BS}+\frac{Q}{L_{BS}})}_{C_1}F+\underbrace{W\zeta(NKp-\frac{Q}{L_{BS}}+KP_{UT}+P_{fix})}_{C_0} =\sum_{i=0}^{3}C_iF^i. \label{PMUSIC}
	\end{align}
	MSE of the MUSIC algorithm is defined as
	\begin{equation}
		MSE=\frac{1}{N_{MC}}\sum_{i=1}^{N_{MC}}\sum_{k=1}^{K}(\theta_k-\hat{\theta}_k)^2, \label{RMSE}
	\end{equation}
	where $N_{MC}$ is number of Monte-Carlo simulations of MUSIC algorithm. 
	
	Finally, we formulate the $LE$ as
	\begin{equation}
		LE_{MUSIC}=\frac{K}{(E_t)(\sqrt{MSE})}, \label{LEMUSIC}
	\end{equation}
	In the next section, the MUSIC algorithm's $LE$ is optimized w.r.t. $F$, and its behavior when these $F$ antennas are selected from the first and furthest set is illustrated. 
	\begin{table*}[t]
		\centering
		\caption{Simulation parameters}
		\begin{tabular}{|c|c||c|c|}
			\hline
			\textbf{Parameter} & \textbf{Value} & \textbf{Parameter} & \textbf{Value}\\\hline
			Bandwidth: $W$ & $50KHz$ & Operational efficiency: $L_{BS}$&$30 (GFLOP/Joule)$ \\
			Pilot transmission time: $\zeta$ & $0.5(ms)$ & BS's RF-chain Power consumption: $P_{BS}$ & $1W$ \\
			Noise variance: $\sigma^2_n$ & $10^{-20} (W/Hz)$ & UT's RF-chain Power consumption: $P_{UT}$ & $0.3W$ \\
			Channel coefficients: $h^r,h^i$ & $\mathcal{N}(0,0.5)$ & Fixed power consumption in BS: $P_{fix}$ & $0.5W$ \\
			Received pilot power: $p$ & $10^{-19} (W/Hz)$ & Antenna separation ratio to wavelength: $\frac{d}{\lambda}$ & $0.5$ \\
			\hline
		\end{tabular}
		\label{SParameter}
	\end{table*}
	\section{Numerical Results} \label{Numerical}
	In this section, we verify our analytical results that are obtained in previous sections. We study the behavior of $LE$ function for different scenarios, using Monte-Carlo simulations when analytical traceability is not possible. Parameters that we use are listed in Table~\ref{SParameter}, unless otherwise stated. 
	\begin{figure}[t]
		\centering
		\includegraphics[width=0.5\textwidth]{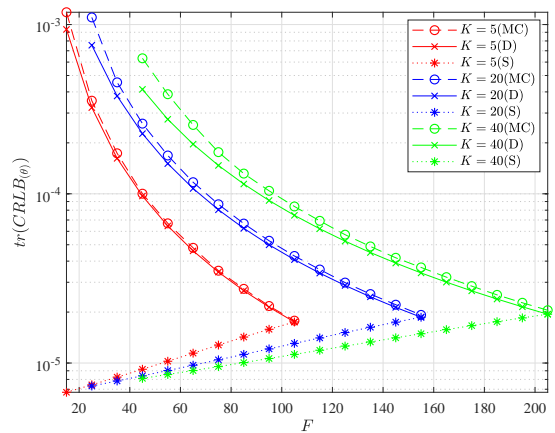}
		\caption{Deterministic and Monte-Carlo simulations of $CRLB_\theta$.}
		\label{FCRLB2}
	\end{figure}
	\begin{figure}[t]
		\centering
		\includegraphics[width=0.5\textwidth]{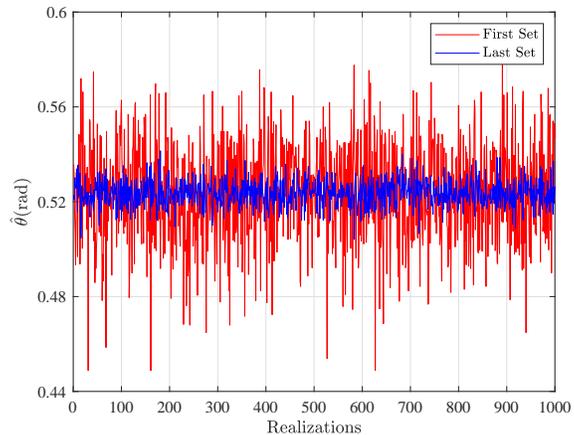}
		\caption{ML estimation for a single UT with $M=16$.}
		\label{MLVariance}
	\end{figure}
	\begin{figure}[t]
		\centering
		\includegraphics[width=0.5\textwidth]{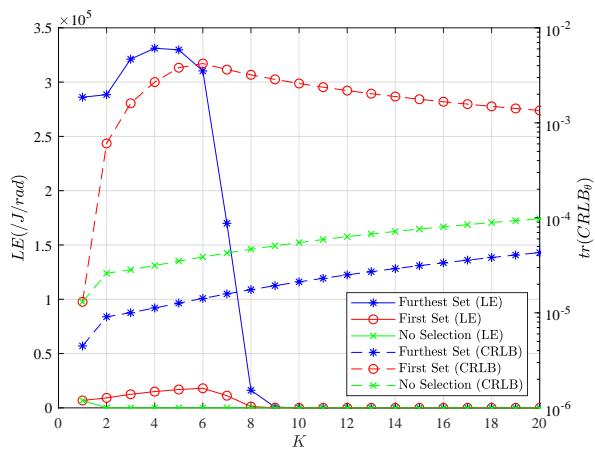}
		\caption{Optimal $LE$ for each corresponding $K$ in ML estimation.}
		\label{MLLE}
	\end{figure}
	
	Fig.~\ref{FCRLB2} shows $tr(\boldsymbol{CRLB}_\theta)$ in which dashed lines are generated by Monte-Carlo simulations (indicated by MC) while solid lines are computed using approximated expression in Eq.~\ref{CRLBASthetaS1} for $K=5,20,40$. As $CRLB_\theta$ grows indefinitely when $\theta\to 0,\pi$, an area of $\pi/10$ from each side is excluded, and UTs are equispaced in the remaining area. Interestingly, deterministic approximation (indicated by D) converges very fast, even when the number of antennas is not so large. Also, Eq.~\ref{CRLBASthetaS1} has the same behavior as Monte-Carlo simulations. It should be noted that as the trace of $\boldsymbol{CRLB}_\theta$ is plotted, the distance between analytical expression and Monte-Carlo simulations will increase for the higher number of UTs since it is the sum of $K$ almost sure convergence. This is why there is a seeming increment between analytical and Monte-Carlo curves in this figure. Furthermore, deterministic $CRLB_{\theta}$ in Eq.\ref{CRLBASthetaopt}, is plotted (indicated by S) when $M=100,155,205$ for corresponding $K=5,20,40$ curves. We see a significant decrease in $CRLB_\theta$ when the furthest set of antennas is used relative to the case when the first group of them is used, more than two orders of magnitude in some cases. As the number of utilized antennas grows, $CRLB_\theta$ for the furthest set grows and becomes closer to the first set curve until all available antennas are used, when they become the same. The larger the $M$, the lower the $CRLB_\theta$ will become for the furthest set.
	\begin{figure}[t]
		\centering
		\includegraphics[width=0.5\textwidth]{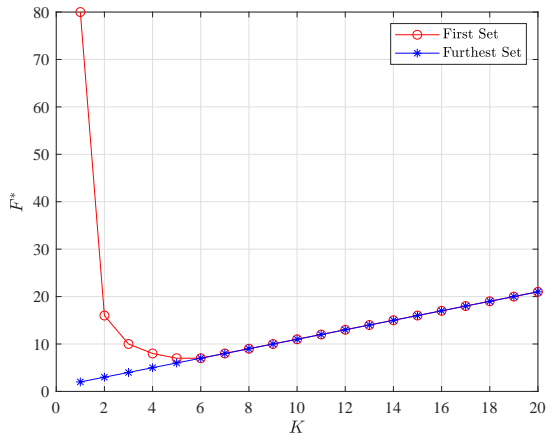}
		\caption{Optimal number of antennas for ML.}
		\label{Mopt}
	\end{figure}
	
	In Fig.~\ref{MLVariance} an ML estimator's output for a single UT is plotted to compare furthest and first antenna selection when the number of antennas is not so large in Monte-Carlo simulations. In this figure, $M=16$ and $F=6$ and UT's AoA is $\theta=\frac{\pi}{6}$. It is seen that the furthest set has dramatically lower variance in its estimation, proving that it outperforms antenna selection from the first set. Ratio of mean variance of furthest antenna set to first set is $0.066$ in Monte-Carlo simulation and the ratio predicted by Eq.~\ref{CRLBASthetaopt} and Eq.~\ref{CRLBASthetaS1} is $0.061$. This proves the accuracy of these deterministic equations.
	
	In Fig.~\ref{MLLE} $LE$ for different numbers of UTs is plotted. In this figure, the right side axis is for LE scale, and solid lines are its corresponding curves. The left side axis is for the $tr(CRLB_\theta)$, and dashed lines are its corresponding curves. For each scenario, the same colors are used. In this figure, $M=80$ and for the green curves (Eq.~\ref{LEp1}), in each $K$, all of the available antennas are used, i.e., no antenna selection. We see that due to high computational complexity, $LE$ drops sharply as $K$ increases. In other two curves (blue for Eq.~\ref{LEl1} and red for replacing $M$ with $F$ in Eq.~\ref{LEp1}), for each $K$, $LE$ is maximized w.r.t. $F$, with constraint $F\ge K$. We see that $LE$ is significantly improved when the optimal number of antennas is used instead of all available antennas. This confirms that using all of the antennas is not always efficient. Moreover, when the optimal number of antennas are selected from the furthest set, $LE$ increases even further, up to $220\%$ in some points, highlighting the advantage of the proposed antenna selection. As $K$ increases, energy consumption increases exponentially, and this causes $LE$ of different scenarios to decrease and become close to each other. However, they are not exactly same until all of the available antennas are used. Furthermore, the $tr(CRLB_\theta)$ indicates the minimum achievable accuracy for each point. One can put a minimum accuracy constraint when optimizing $LE$ depending on one's requirements.
	
	Fig.~\ref{Mopt} shows corresponding optimal number of antennas for each $K$ that maximize $LE$ in Fig.~\ref{MLLE}. It is seen that due to the exponential growth of energy consumption, $F^*$ decreases very fast for the first set and reaches saturation point ($F^*=K+1$) as $K$ increases. Accordingly, this is the point after which $LE$ starts to decrease in Fig.~\ref{MLLE}. Also, as predicted by Eq.~\ref{Fopt}, the optimal number of antennas for the furthest set is always the minimum possible number of antennas, $K+1$. This illustrates that furthest antenna selection obtains higher $LE$ using fewer antennas, which reduces costs of construction and maintenance of the system. 
	
	\begin{figure*}[t!]
		\centering
		\begin{subfigure}[t]{0.4\textwidth}
			\centering
			\includegraphics[width=1.05\textwidth]{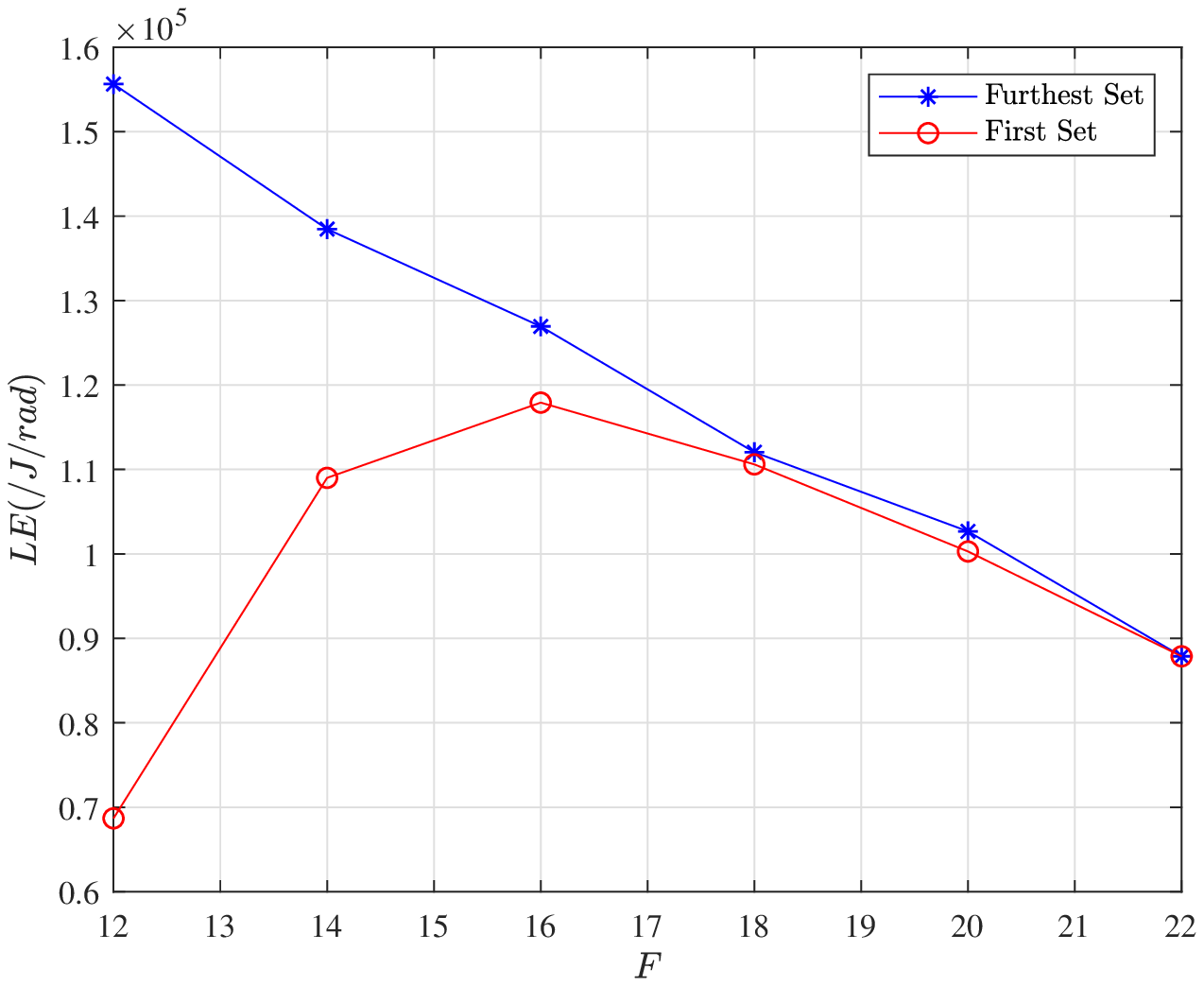}
			\caption{$K=10$ and $M=22$.}
		\end{subfigure}%
		~ 
		\begin{subfigure}[t]{0.4\textwidth}
			\centering
			\includegraphics[width=1.05\textwidth]{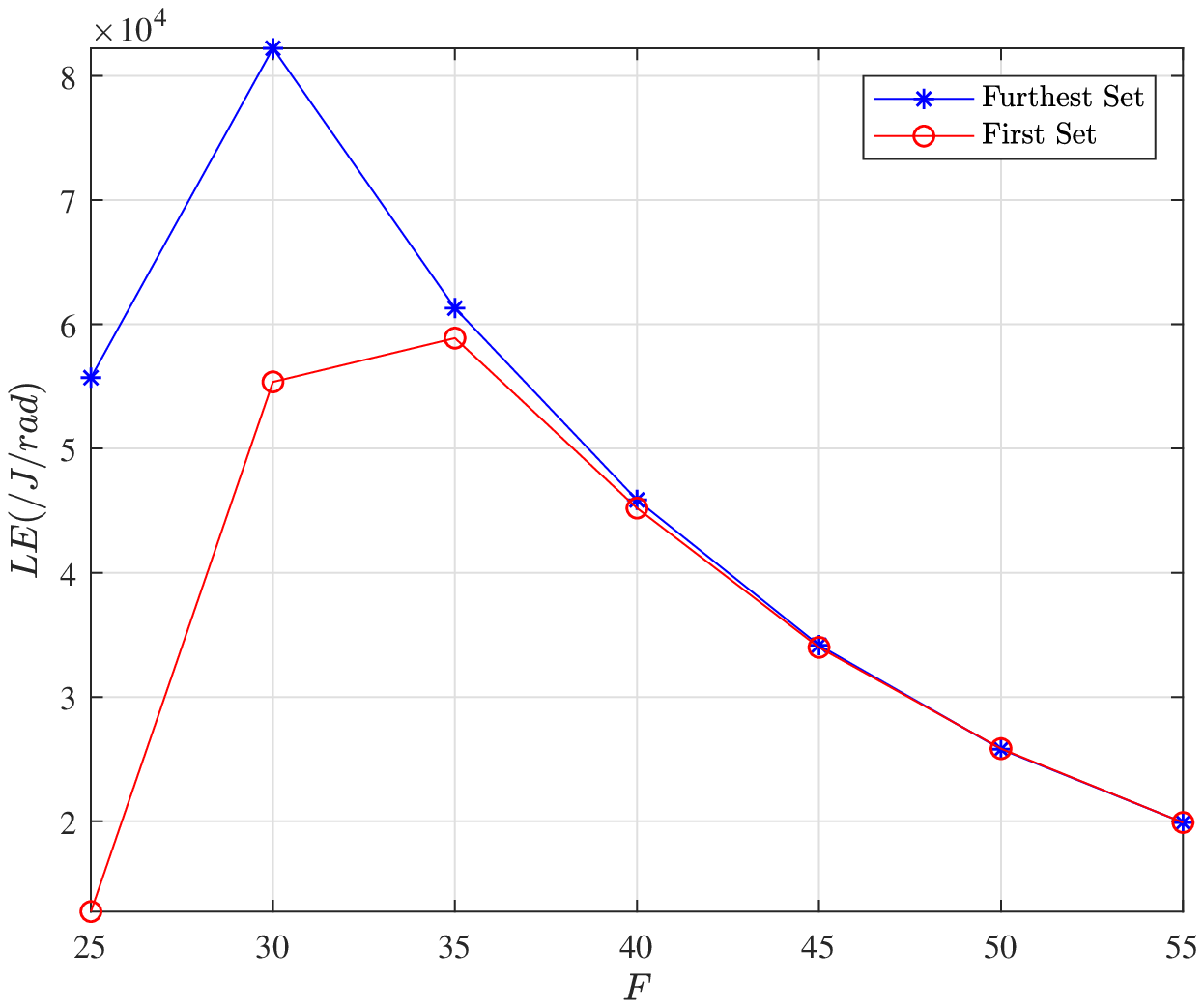}
			\caption{$K=20$ and $M=55$.}
		\end{subfigure}
		\hfill
		\centering
		\begin{subfigure}[t]{0.4\textwidth}
			\centering
			\includegraphics[width=1.05\textwidth]{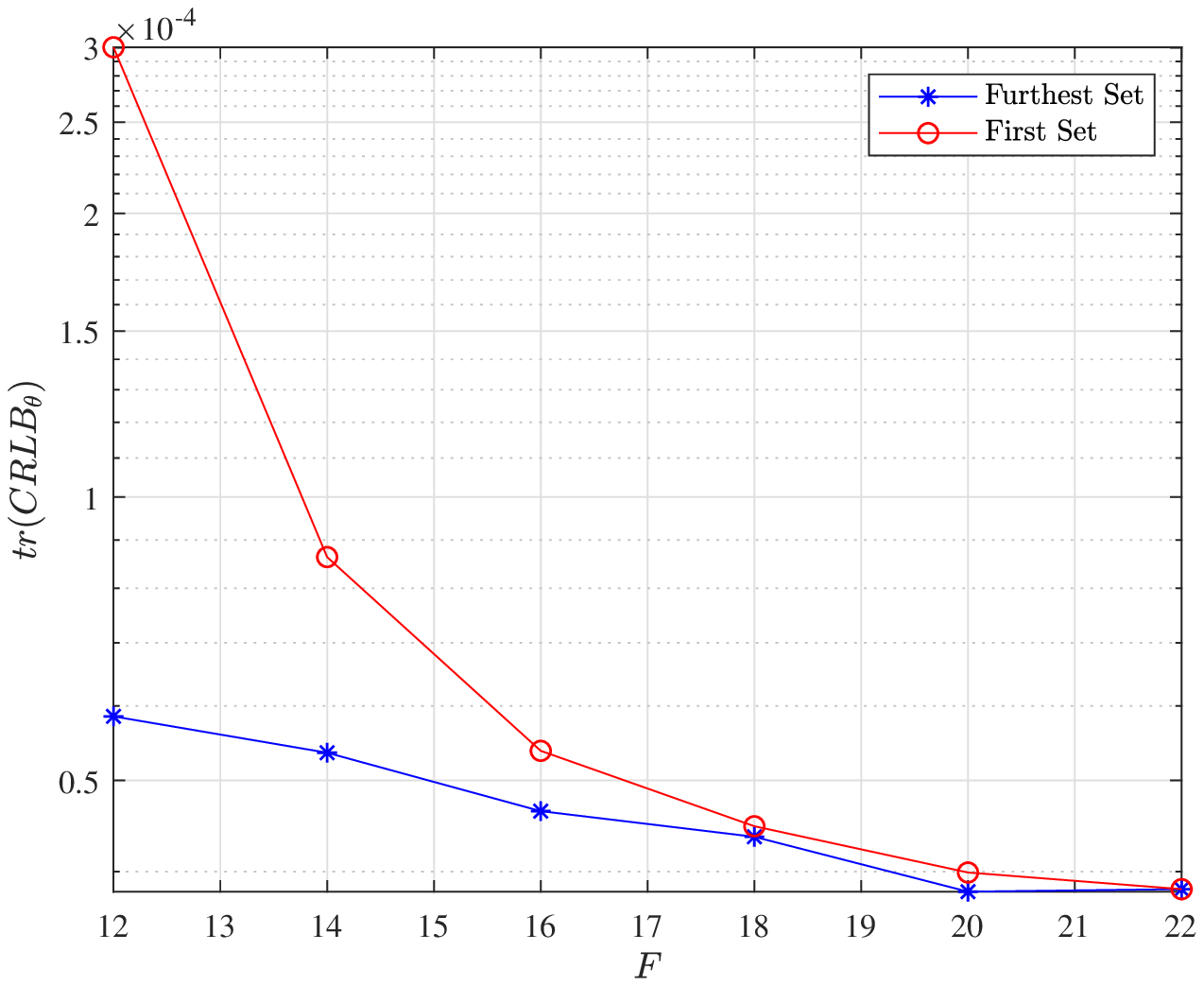}
			\caption{$K=10$ and $M=22$.}
		\end{subfigure}%
		~ 
		\begin{subfigure}[t]{0.4\textwidth}
			\centering
			\includegraphics[width=1.05\textwidth]{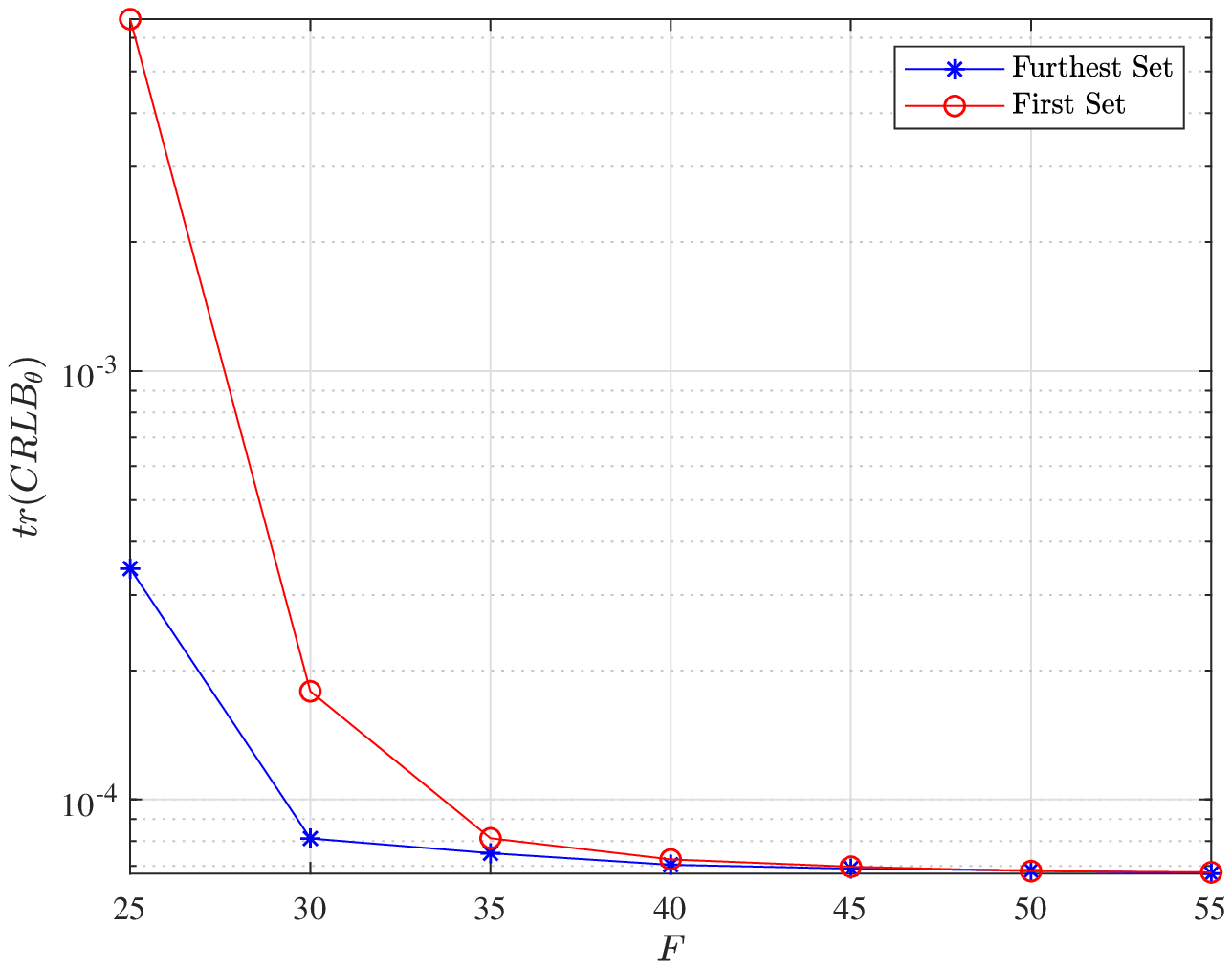}
			\caption{$K=20$ and $M=55$.}
		\end{subfigure}
		\caption{$LE$ and $tr(\boldsymbol{CRLB}_\theta)$ of MUSIC algorithm versus $F$.}\label{LEMUSICf}
	\end{figure*}
	
	$LE$ of MUSIC algorithm is plotted by Monte-Carlo simulation of Fig.~\ref{LEMUSICf}. The left hand side sub-figures of Fig.~\ref{LEMUSICf} are for $K=10$ and $M=22$, and the right hand side sub-figures are for $K=20$ and $M=55$. Same as Fig.~\ref{FCRLB2}, an exclusion area of $\pi/7$ is considered here. By selecting antennas from the first set, $LE$ has an optimum point in $F^*=16$ and $F^*=35$ in these settings, respectively. On the other hand, for furthest antenna selection, in $K=10$, $LE$ is always decreasing, and so its optimum point is at minimum number of antennas, but in $K=20$ $LE$ of the furthest set has an optimum point in $F=30$. Nevertheless, its optimum point always happens before the first set's optimum point, meaning that the furthest antenna selection needs fewer antennas in this scenario, as well. In addition, the furthest antenna selection always has higher $LE$ than the first antenna selection proving that furthest antenna selection is always beneficial, no matter what the channel model is. Furthermore, it is seen that the obtained $\boldsymbol{CRLB}_\theta$ for the furthest set is always lower than the first set, proving the superiority of the furthest antenna selection, once more.
	
	\section{Conclusion} \label{Conclusion}
	This paper analyzed $CRLB$ for AoA estimation when the received signal at the BS is accompanied by multipath signals according to the DMC model. With the help of RMT, we proved that regardless of the distribution of multipath signals, $CRLB$ for AoA almost surely converges toward a closed-form expression in the MU-Massive MIMO setting. This illustrates the contribution of multipath signals in the AoA estimation, providing a theoretical basis for recent studies that estimated location information with the help of multipath signals.
	
	A refined version of the localization efficiency function is presented, which is a ratio of benefits to costs in localization and reflects the trade-off between performance and energy consumption. Contrary to previous studies, we used a comprehensive energy consumption model in this function and showed that there is an optimal number of antennas that maximizes the efficiency in the localization phase. We presented an antenna selection method that minimizes $CRLB$ for AoA when the number of utilized antennas is smaller than the total available antennas. Also, the behavior of both $CRLB$ and $LE$ for this selection scheme is studied. It is shown that the behaviors of both of them change when utilized antennas are selected based on the proposed scheme, which affects the optimal number of antennas that maximizes $LE$.
	
	Numerical results confirmed the $CRLB$'s convergence, even when the number of antennas is not too large. They showed that the proposed antenna selection strategy dramatically reduces $CRLB$. This phenomenon has been validated by Monte-Carlo simulations, too. Furthermore, simulation results confirmed significant improvement of $LE$ when our antenna selection approach is utilized. In fact, with the help of the proposed antenna selection method, ML estimation gains a competitive advantage over a certain region in terms of efficiency. In the end, $LE$ of the MUSIC algorithm is studied and simulated, indicating the applicability of our antenna selection strategy even when the contribution of multipath signals is ignored. 
	\section{acknowledgment}
	This work is supported by F.R.S.-FNRS under the EOS research project MUSEWINET (EOS project 30452698).
	\bibliographystyle{ieeetr}
	
	\appendices
	\section{Proof of Lemma \ref{Lem3}} \label{Appl3}
	By defining $\gamma=\beta(\cos(\theta_1)-\cos(\theta_2))$, from \cite{stoica1989maximum} we know that when $\gamma\ne 0$
	\begin{equation}
		\lim_{M\to\infty}\underbrace{\frac{1}{M^3}\sum_{m=1}^{M}m^2\cos(m\gamma)}_{\triangleq V(M,\gamma)}=0. 
	\end{equation}
	Also, we know that $-2\beta<\gamma<2\beta$, therefore by defining the set $A$ as
	\begin{equation}
		A=\{\gamma\in(-2\beta,2\beta):\lim_{M\to\infty}V(M,\gamma)=0\} 
	\end{equation}
	we have
	\begin{equation}
		P(A)=P(\cos(\theta_1)\ne\cos(\theta_2))
	\end{equation}
	By noting that $\theta\thicksim U[0,\pi]$, we know that the probability of two users having exact same AoA is zero, therefore, $P(A)=1$. So, as $M\to\infty$, $V(M,\gamma)$ is zero with probability one, that is the definition of the almost sure convergence. 
	\qed
	\section{Proof of Theorem \ref{T1}} \label{App1}
	Generally, $\boldsymbol{X}$ is a complex matrix, composed from two real and imaginary parts, so we can write it as 
	\begin{equation}
		\boldsymbol{X}=M\boldsymbol{\Sigma}^{\frac{1}{2}}(\boldsymbol{A}+j\boldsymbol{B})\boldsymbol{D}, 
		\label{a+b}
	\end{equation}  
	where
	\begin{align}
		&\boldsymbol{\Sigma}=\beta^2diag(0,\frac{1}{M^2},\frac{4}{M^2},\ldots,\frac{(M-1)^2}{M^2}),\\
		&\boldsymbol{D}=diag(\sin^2(\theta_1)|s_1|\sqrt{l(r_1)},\ldots,\sin^2(\theta_K)|s_K|\sqrt{l(r_K)}),
		\label{Sig}
	\end{align}
	As $\boldsymbol{\Sigma}$ is a diagonal matrix, its singular values are its diagonal elements. From Eq.~\ref{Sig}, the largest singular value of $\boldsymbol{\Sigma}$ is equal to $1$. Therefore, the spectral norm of $\boldsymbol{\Sigma}$ is bounded \cite{boyd2004convex}. If we replace Eq.~\ref{a+b} in $\boldsymbol{X}^H\boldsymbol{X}$, we have 
	\begin{align}
		\boldsymbol{X}^H\boldsymbol{X}&=M^2\boldsymbol{D}(\boldsymbol{A}+j\boldsymbol{B})^H\boldsymbol{\Sigma}(\boldsymbol{A}+j\boldsymbol{B})\boldsymbol{D}=M^2\boldsymbol{D}(\boldsymbol{A}^T\boldsymbol{\Sigma}\boldsymbol{A}+\boldsymbol{B}^T\boldsymbol{\Sigma}\boldsymbol{B})\boldsymbol{D}\nonumber\\&+jM^2\boldsymbol{D}(\boldsymbol{A}^T\boldsymbol{\Sigma}\boldsymbol{B}-\boldsymbol{B}^T\boldsymbol{\Sigma}\boldsymbol{A})\boldsymbol{D}. \label{}
	\end{align}
	Therefore
	\begin{equation}
		\mathcal{R}e(\boldsymbol{X}^H\boldsymbol{X})=M^2\boldsymbol{D}(\boldsymbol{A}^T\boldsymbol{\Sigma}\boldsymbol{A}+\boldsymbol{B}^T\boldsymbol{\Sigma}\boldsymbol{B})\boldsymbol{D}. \label{Reform}
	\end{equation}
	We can decompose $(\boldsymbol{A}^T\boldsymbol{\Sigma}\boldsymbol{A}+\boldsymbol{B}^T\boldsymbol{\Sigma}\boldsymbol{B})$ as
	\begin{equation}
		\boldsymbol{A}^T\boldsymbol{\Sigma}\boldsymbol{A}+\boldsymbol{B}^T\boldsymbol{\Sigma}\boldsymbol{B}=(\bar{\boldsymbol{A}}^T\boldsymbol{\Sigma}\bar{\boldsymbol{A}}+\bar{\boldsymbol{B}}^T\boldsymbol{\Sigma}\bar{\boldsymbol{B}})+\sigma_h^2(\hat{\boldsymbol{A}}^T\boldsymbol{\Sigma}\hat{\boldsymbol{A}}+\hat{\boldsymbol{B}}^T\boldsymbol{\Sigma}\hat{\boldsymbol{B}}), \label{reform2}
	\end{equation}
	in which 
	\begin{align}
		&\bar{\boldsymbol{A}}_{m,k}=\mathcal{R}e\{\bar{h}_ke^{-j(m\beta \cos(\theta_k)+\varphi_k)}\}, \hspace{4mm} \bar{\boldsymbol{B}}_{m,k}=\mathcal{I}m\{\bar{h}_ke^{-j(m\beta \cos(\theta_k)+\varphi_k)}\},\\
		&\hat{\boldsymbol{A}}_{m,k}=\frac{1}{\sqrt{\sigma_h^2}}\mathcal{R}e\{\hat{h}_{m,k}e^{-j(m\beta \cos(\theta_k)+\varphi_k)}\},\hspace{4mm} \hat{\boldsymbol{B}}_{m,k}=\frac{1}{\sqrt{\sigma_h^2}}\mathcal{I}m\{\hat{h}_{m,k}e^{-j(m\beta \cos(\theta_k)+\varphi_k)}\},
	\end{align}
	and $\varphi_k$ is phase of $k$th transmitted pilot. So, Eq.\ref{Reform} can be written as
	\begin{equation}
		\mathcal{R}e(\boldsymbol{X}^H\boldsymbol{X})=M^2\boldsymbol{D}(\bar{\boldsymbol{A}}^T\boldsymbol{\Sigma}\bar{\boldsymbol{A}}+\bar{\boldsymbol{B}}^T\boldsymbol{\Sigma}\bar{\boldsymbol{B}})\boldsymbol{D}+M^2\sigma_h^2\boldsymbol{D}(\hat{\boldsymbol{A}}^T\boldsymbol{\Sigma}\hat{\boldsymbol{A}}+\hat{\boldsymbol{B}}^T\boldsymbol{\Sigma}\hat{\boldsymbol{B}})\boldsymbol{D}. \label{separate}
	\end{equation}
	Now, we investigate the behavior of both parts of the right hand side of Eq.~\ref{separate} separately. First, we analyze $\bar{\boldsymbol{A}}^T\boldsymbol{\Sigma}\bar{\boldsymbol{A}}+\bar{\boldsymbol{B}}^T\boldsymbol{\Sigma}\bar{\boldsymbol{B}}$. The $(k_1,k_2)$th element of this matrix will be
	\begin{align}
		(\bar{\boldsymbol{A}}^T\boldsymbol{\Sigma}\bar{\boldsymbol{A}}+\bar{\boldsymbol{B}}^T\boldsymbol{\Sigma}\bar{\boldsymbol{B}})_{k_1,k_2}&=\mathcal{R}e\{\bar{h}_{k_1}\boldsymbol{a}_{Rx}(\theta_{k_1})^T\}\boldsymbol{\Sigma}\mathcal{R}e\{\bar{h}_{k_2}\boldsymbol{a}_{Rx}(\theta_{k_2})\}\nonumber\\&+\mathcal{I}m\{\bar{h}_{k_1}\boldsymbol{a}_{Rx}(\theta_{k_1})^T\}\boldsymbol{\Sigma}\mathcal{I}m\{\bar{h}_{k_2}\boldsymbol{a}_{Rx}(\theta_{k_2})\}.
	\end{align}
	For any two random variables $p$ and $q$, we have
	\begin{equation} 
		\mathcal{R}e\{pq^*\}=\mathcal{R}e\{p\}\mathcal{R}e\{q\}+\mathcal{I}m\{p\}\mathcal{I}m\{q\}.
	\end{equation}
	So
	\begin{align}
		(\bar{\boldsymbol{A}}^T\boldsymbol{\Sigma}\bar{\boldsymbol{A}}+\bar{\boldsymbol{B}}^T\boldsymbol{\Sigma}\bar{\boldsymbol{B}})_{k_1,k_2}&=\mathcal{R}e\{\bar{h}_{k_1}^*\bar{h}_{k_2}e^{j(\varphi_{k_1}-\varphi_{k_2})}\boldsymbol{a}_{Rx}(\theta_{k_1})^H\boldsymbol{\Sigma}\boldsymbol{a}_{Rx}(\theta_{k_2})\}.
	\end{align}
	Using Lemma \ref{Lem3}, we obtain 
	\begin{equation} 
		\bar{\boldsymbol{A}}^T\boldsymbol{\Sigma}\bar{\boldsymbol{A}}+\bar{\boldsymbol{B}}^T\boldsymbol{\Sigma}\bar{\boldsymbol{B}}\xrightarrow{a.s.}\frac{tr(\boldsymbol{\Sigma})}{M}diag(|\bar{h}_{1}|^2,\ldots, |\bar{h}_{K}|^2)\label{1rhs}
	\end{equation}
	Now, we analyze the second part of the right hand side of Eq.~\ref{reform2}. Choosing an arbitrary element of $\hat{\boldsymbol{A}}$, we have
	\begin{align}
		\hat{\boldsymbol{A}}_{m,k}&=\frac{1}{\sqrt{M\sigma_h^2}}\mathcal{R}e\{h_{m,k}e^{-j(m\beta \cos(\theta_k)+\varphi_k)}\}=\frac{1}{\sqrt{M\sigma_h^2}}(h^r_{m,k}\cos(m\beta\cos(\theta_k)+\varphi_k)\nonumber\\&+h^i_{m,k}\sin(m\beta\cos(\theta_k)+\varphi_k)).
	\end{align}
	So the variance of each element will be
	\begin{align}
		\mathbb{E}\{\hat{\boldsymbol{A}}_{m,k}\hat{\boldsymbol{A}}_{m,k}^*\}&=\frac{1}{M\sigma_h^2} \mathbb{E}\{|h^r_{m,k}|^2\}\cos^2((m-1)\sqrt{\beta}\cos(\theta_k)+\varphi_k)\nonumber\\&+\frac{1}{M\sigma_h^2}\mathbb{E}\{|h^i_{m,k}|^2\}\sin^2((m-1)\sqrt{\beta}\cos(\theta_k)+\varphi_k)=\frac{1}{M}.
	\end{align}
	Similarly, $\mathbb{E}\{\hat{\boldsymbol{B}_{m,k}}\hat{\boldsymbol{B}_{m,k}}^*\}=\frac{1}{M}$. For the fourth order moment we have
	\begin{align}
		&\mathbb{E}\{|\hat{\boldsymbol{A}}_{m,k}|^4\}=\mathbb{E}\{\hat{\boldsymbol{A}}_{m,k}\hat{\boldsymbol{A}}_{m,k}^*\hat{\boldsymbol{A}}_{m,k}\hat{\boldsymbol{A}}_{m,k}^*\}\nonumber\\&=\frac{1}{M^2\sigma_h^4}\mathbb{E}\{(h^r_{m,k}\cos((m-1)\beta\cos(\theta_k)+\varphi_k)+h^i_{m,k}\sin((m-1)\beta\cos(\theta_k)+\varphi_k))^4\}
	\end{align}
	After some algebraic simplification and using the fact that $h^r$ and $h^i$ are independent from each other, we obtain the following term for the fourth order moment
	\begin{align}
		\mathbb{E}\{|\hat{\boldsymbol{A}}_{m,k}|^4\}&=\frac{1}{M^2\sigma_h^4}[\mathbb{E}\{|h^r_{m,k}|^4\}\cos^4((m-1)\beta\cos(\theta_k)+\varphi_k)\nonumber\\&+\mathbb{E}\{|h^i_{m,k}|^4\}\sin^4((m-1)\beta\cos(\theta_k)+\varphi_k)]\nonumber\\&+\frac{6}{M^2}\cos^2((m-1)\beta\cos(\theta_k)+\varphi_k)\sin^2((m-1)\beta\cos(\theta_k)+\varphi_k). \label{4mom}
	\end{align}
	Noting the fact that channel distribution is independent from the number of antennas at the BS, from Eq.~\ref{4mom} it is evident that for a bounded $\mathbb{E}\{|h^r_{m,k}|^4\}$ and $\mathbb{E}\{|h^i_{m,k}|^4\}$, fourth moment of all elements of $\boldsymbol{A}$ are in the order of $\mathcal{O}(\frac{1}{M^2})$. Same condition is held for all the elements of $\boldsymbol{B}$. Using same approach, it can be shown that the order of eighth moment of all the elements of both $\boldsymbol{A}$ and $\boldsymbol{B}$ are in the order of $\mathcal{O}(\frac{1}{M^4})$.
	
	Rewriting $\hat{\boldsymbol{A}}$ as
	\begin{equation}
		\hat{\boldsymbol{A}}=[\hat{\boldsymbol{a}}_1 \hat{\boldsymbol{a}}_2 \ldots \hat{\boldsymbol{a}}_K].
	\end{equation}
	From Lemma~\ref{Lem1} and Lemma~\ref{Lem2} we have
	\begin{align}
		&\hat{\boldsymbol{a}}^H_1\boldsymbol{\Sigma}\hat{\boldsymbol{a}}_1-\frac{1}{M}tr(\boldsymbol{\Sigma})\xrightarrow{a.s.}0, \nonumber \\
		&\hat{\boldsymbol{a}}^H_2\boldsymbol{\Sigma}\hat{\boldsymbol{a}}_2-\frac{1}{M}tr(\boldsymbol{\Sigma})\xrightarrow{a.s.}0, \nonumber \\
		&\hat{\boldsymbol{a}}^H_1\boldsymbol{\Sigma}\hat{\boldsymbol{a}}_2\xrightarrow{a.s.}0, \nonumber \\
		&\hat{\boldsymbol{a}}^H_2\boldsymbol{\Sigma}\hat{\boldsymbol{a}}_1\xrightarrow{a.s.}0. \nonumber
	\end{align}
	So, we can write
	\begin{equation}
		\left[\begin{array}{c} 
			\hspace{2mm}\hat{\boldsymbol{a}}^H_1 \hspace{2mm}\\
			\hspace{2mm}\hat{\boldsymbol{a}}^H_2 \hspace{2mm}\\
		\end{array}\right] \boldsymbol{\Sigma} \left[\begin{array}{c} 
			\hat{\boldsymbol{a}}_1 \hspace{2mm} \hat{\boldsymbol{a}}_2
		\end{array}\right]- \frac{1}{M}\left[\begin{array}{cc} 
			tr(\boldsymbol{\Sigma}) & 0 \\
			0 & tr(\boldsymbol{\Sigma})
		\end{array}\right]\xrightarrow{a.s.}\boldsymbol{0}_2.
	\end{equation}
	After repeating this process for $K$ columns of $\boldsymbol{A}$, we will have
	\begin{equation}
		\hat{\boldsymbol{A}}^H\boldsymbol{\Sigma}\hat{\boldsymbol{A}}-\frac{1}{M}tr(\boldsymbol{\Sigma})\boldsymbol{I}_K\xrightarrow{a.s.}0. \label{traceAbar}
	\end{equation}
	Similarly, 
	\begin{equation}
		\hat{\boldsymbol{B}}^H\boldsymbol{\Sigma}\hat{\boldsymbol{B}}-\frac{1}{M}tr(\boldsymbol{\Sigma})\boldsymbol{I}_K\xrightarrow{a.s.}0. \label{traceBbar}
	\end{equation}
	Based on \ref{traceAbar} and \ref{traceBbar} and the fact that both $\boldsymbol{A}$ and $\boldsymbol{B}$ consist of real elements, we can write
	\begin{equation}
		\hat{\boldsymbol{A}}^T\boldsymbol{\Sigma}\hat{\boldsymbol{A}}+\hat{\boldsymbol{B}}^T\boldsymbol{\Sigma}\hat{\boldsymbol{B}}-\frac{2}{M}tr(\boldsymbol{\Sigma})\boldsymbol{I}_K\xrightarrow{a.s.}0. \label{AtABtB}
	\end{equation}
	By replacing Eq.~\ref{1rhs} and Eq.~\ref{AtABtB} in Eq.~\ref{reform2}, we obtain
	\begin{equation}
		\boldsymbol{A}^T\boldsymbol{\Sigma}\boldsymbol{A}+\boldsymbol{B}^T\boldsymbol{\Sigma}\boldsymbol{B}\xrightarrow{a.s.}\frac{tr(\boldsymbol{\Sigma})}{M}(2\sigma_h^2+diag(|\bar{h}_{1}|^2,\ldots, |\bar{h}_{K}|^2)).
	\end{equation}
	Consequently
	\begin{equation}
		(\boldsymbol{A}^T\boldsymbol{\Sigma}\boldsymbol{A}+\boldsymbol{B}^T\boldsymbol{\Sigma}\boldsymbol{B})^{-1}\xrightarrow{a.s.}\frac{M}{2tr(\boldsymbol{\Sigma})}(2\sigma_h^2+diag(|\bar{h}_{1}|^2,\ldots, |\bar{h}_{K}|^2))^{-1}. \label{sdsd}
	\end{equation}
	Defining $\boldsymbol{S}\triangleq\sigma_n^2\boldsymbol{D}^{-2}$ and noticing that
	\begin{equation}
		tr(\boldsymbol{\Sigma})=\frac{\beta^2}{M^2}\sum_{m=0}^{M-1}m^2=\frac{\beta^2M(M-1)(2M-1)}{6M^2}, \label{alph2}
	\end{equation}
	we obtain
	\begin{align}
		\boldsymbol{CRLB}_\theta=\frac{\sigma_n^2}{2}(\mathcal{R}e(\boldsymbol{X}^H\boldsymbol{X}))^{-1}\xrightarrow{a.s.}\frac{3}{\beta^2(M-1)(2M-1)}\boldsymbol{S}. \label{traceReal}
	\end{align}
	\qed
\end{document}